\documentclass{article}
\usepackage{fullpage}
\usepackage[utf8]{inputenc}
\usepackage{hyperref}
\usepackage{xcolor}
\usepackage{url}
\usepackage{amsmath,amsthm,amsfonts,amssymb,microtype,bbm}
\usepackage{graphicx}
\usepackage{subfig}
\usepackage{dsfont}
\usepackage{natbib}
\usepackage{bm}
\usepackage{mathtools}
\usepackage{authblk}

\newtheorem{remark}{Remark}
\newtheorem{lemma}{Lemma}
\newtheorem{claim}{Claim}

\newtheorem{fact}{Fact}
\newtheorem{example}{Example}

\newtheorem{definition}{Definition}
\newtheorem{corollary}{Corollary}

\newtheorem{theorem}{Theorem}
\newtheorem{assumption}{Assumption}

\newcommand{\reals}{\mathbb{R}}
\newcommand{\cD}{\mathcal{D}}
\newcommand{\cN}{\mathcal{N}}
\newcommand{\cW}{\mathcal{W}}
\newcommand{\E}{\mathbb{E}}
\renewcommand{\Pr}{\mathbb{P}}
\newcommand{\pdv}[2]{\frac{\partial}{\partial #2} #1}

\newcommand{\thr}{\tau}
\newcommand{\stdw}{\sigma}
\newcommand{\stdt}{\gamma}
\newcommand{\discnt}{\rho}

\newcommand{\esh}[1]{\textcolor{red}{[Eshwar: #1]}}
\newcommand{\juba}[1]{\textcolor{green}{[Juba: #1]}}
\newcommand{\kri}[1]{\textcolor{purple}{[Krishna: #1]}}
\newcommand{\ar}[1]{\textcolor{blue}{[Aaron: #1]}}
\newcommand{\sampk}[1]{\textcolor{orange}{[Sampath: #1]}}

\renewcommand{\esh}[1]{}
\renewcommand{\juba}[1]{}
\renewcommand{\kri}[1]{}
\renewcommand{\ar}[1]{}
\renewcommand{\sampk}[1]{}

\title{Wealth Dynamics Over Generations: \\ Analysis and Interventions}
\author[1]{Krishna Acharya}
\author[2]{Eshwar Ram Arunachaleswaran}
\author[2]{Sampath Kannan}
\author[2]{Aaron Roth}
\author[1]{Juba Ziani}

\affil[1]{Georgia Tech ISyE}
\affil[2]{University of Pennsylvania Department of Computer and Information Sciences}
\begin{document}

\maketitle

\begin{abstract}
We present a stylized model with feedback loops for the evolution of a population's wealth over generations. Individuals have both talent and wealth: talent is a random variable distributed identically for everyone, but wealth is a random variable that is dependent on the population one is born into. Individuals then apply to a downstream agent, which we treat as a university throughout the paper (but could also represent an employer) who makes a decision about whether to admit them or not. The university does not directly observe talent or wealth, but rather a signal (representing e.g. a standardized test) that is a convex combination of both. The university knows the distributions from which an individual's type and wealth are drawn, and makes its decisions based on the posterior distribution of the applicant's characteristics conditional on their population and signal. Each population's wealth distribution at the next round then depends on the fraction of that population that was admitted by the university at the previous round.

We study wealth dynamics in this model, and give conditions under which the dynamics have a single attracting fixed point (which implies population wealth inequality is transitory), and conditions under which it can have multiple attracting fixed points (which implies that population wealth inequality can be persistent). In the case in which there are multiple attracting fixed points, we study interventions aimed at eliminating or mitigating inequality, including increasing the capacity of the university to admit more people, aligning the signal generated by individuals with the preferences of the university, and making direct monetary transfers to the less wealthy population. 

\end{abstract}

\section{Introduction}

The wealth of a population evolves over generations as a function of the opportunities available to it. Opportunities available to a generation  depend not only on their talent, but also on the wealth of the previous generation. In such a dynamical system, the initial wealth of a population determines how wealth evolves and what it will be in the limit. Understanding this system can help illuminate when and why inequalities can arise and persist. 

In this paper we define and analyze a simple, mathematically-tractable model for this feedback system, before considering possible interventions to make its behavior more equitable.  To discuss the main conclusions of our paper, we first need to provide a sketch of our model. Individuals are divided across multiple \emph{populations}, and have both a \emph{type} (an abstraction of talent) and a \emph{wealth}. Within a single population, the distribution of wealth and types are given by Gaussians with known means and variances. Types are distributed identically across populations, but each population has its own distribution of wealth. An individual from a particular population is sampled by sampling their type $T$ from the (universal) type distribution, and their wealth $W$ from the wealth  distribution particular to their population. An individual then generates a \emph{signal} $S = \beta T + (1-\beta) W$, i.e., some convex combination of their wealth and type. This signal could represent e.g. an individual's score on a standardized test, or the rating that results from an interview. Here we allow that the signal might have a dependence on wealth rather than just type because of the indirect effects it can have on evaluations: for example, the ability to engage in additional test preparation.  Downstream, a university\footnote{Throughout this paper we describe the downstream agent as a university admitting students. However we could also view the downstream agent as an employer hiring employees, or any other agent allocating opportunities based on evidence that conflates talent and wealth that have effects on the long-term wealth of the selected individuals.}  who has no direct knowledge of applicant's types or wealth (but with knowledge of the distributions from which they were drawn) observes the signal, and forms a posterior belief about the applicant's type and wealth. The university seeks to select individuals for whom another convex combination $\alpha T + (1 - \alpha) W$
exceeds some threshold $\thr$, and so selects exactly those applicants for whom $\E[\alpha T + (1 - \alpha) W|S] \geq \thr$. Here again we allow that the university might have an explicit preference for wealth (and not purely for type). This might represent e.g. a desire for full tuition payments or future alumni donations, or a more nebulous desire for ``culture fit'' or for skills associated with wealth (e.g. students who can walk on to the sailing or squash team). For each population we then let the mean wealth of the next generation be a non-decreasing function of the fraction of people admitted to the university at the previous round. We also assume that the distribution of types (or talent) remains unchanged over generations and is identical for all individuals, independent of their population.

First, we consider the fixed points of these wealth dynamics. If there is only a single fixed point (and the dynamics converge to it), this implies that wealth inequality across population groups is transitory, and that over time it will equalize (as the mean wealth of all populations  moves to the single attracting fixed point). On the other hand, if there are multiple fixed points of the wealth dynamics, then wealth inequality can persist, with different populations ``stuck'' at different fixed points. 

We give conditions under which the dynamics correspond to a contraction map and have a single fixed point (implying that wealth inequality is transitory). These conditions in particular include the case when $\alpha = 1$ --- i.e. when the university is selecting entirely based on inferred type. On the other hand, there are other situations (in which, necessarily the university places some weight $(1-\alpha) > 0$ on wealth) in which case there can be two attracting fixed points (and a third unstable fixed point), which can result in persistent inequality absent intervention: one population can be ``trapped'' in the less wealthy fixed point, while the other one is in the more wealthy fixed point.

We then turn our attention to interventions. We focus our study of interventions on ways to move a population's wealth from the lower fixed point to the higher fixed point, or to modify the dynamics so that there is a single attracting fixed point (which leads to wealth equality). We consider three types of interventions:
\begin{enumerate}
    \item  \textbf{Increasing The Capacity of the University}: We consider what happens when the university is able to admit more applicants (by lowering its threshold $\thr$). We show that doing this has positive effects: either it shifts the dynamics from the regime in which there are multiple fixed points to the regime in which there is a single fixed point (thus leading to long-term wealth equality), or it \emph{raises the wealth of both attracting fixed points}. 
    \item \textbf{Changing the Design of the Signal $S$}: We consider what happens if we are able to better align the signal the university receives with the university's objective function (by shifting $\beta$ closer to $\alpha$ --- i.e. by having the signal weight type and wealth more similarly to how they are weighted in the university's objective function). We show that as $\beta$ is moved closer to $\alpha$ the disparity between the two fixed points is reduced. Notably, and perhaps counter-intuitively, making the signal depend more on type (by increasing $\beta$) is \emph{not} always the way to reduce disparities (despite the fact that type is distributed identically across populations). 
    \item \textbf{Direct Subsidies to the Disadvantaged Population}: Finally we consider making direct financial subsidies to the disadvantaged population, to shift them from the lower wealth attracting fixed point to the basin of attraction of the higher wealth fixed point (from which they will naturally proceed to the higher wealth fixed point without further intervention). We consider a parameterized family of objective functions that the designer might have, that differ in how they relatively weight the \emph{cost} of the subsidy with the wealth of the disadvantaged population, and in how they discount time. Within this class of interventions, we focus on two options: the most aggressive ``1-shot'' option makes a large 1-shot payment to directly increase the wealth of the disadvantaged population to move them to the basin of attraction of the wealthier fixed point. The least aggressive ``limiting'' option makes the minimal payment per round that is guaranteed to cause eventual convergence to the wealthier fixed point. We derive conditions under which the ``1-shot'' option is preferred by the designer over the ``limiting'' option and vice versa.
    
    
\end{enumerate}

\subsection{Discussion and Limitations}
For mathematical tractability, we study a simple stylized model, which should be viewed as a first cut at attempting to model wealth inequality rather than a faithful description of the full problem. For example, we have assumed that the university has access to an applicant's wealth only indirectly via inferences that can be drawn from their test score and population. In practice, a university has  a number of other signals at their disposal. One should interpret the wealth populations in our model as equivalence classes induced by the information available to them at admissions time. \sampk{Not sure what the previous sentence means... needs clarification.} Similarly, we have modeled individual talent via a static ``type'' distribution, when in fact talent is multi-dimensional and not static (and might \emph{depend} on opportunities that different populations might have different access to prior to university admissions). We have not modelled university capacity constraints, and this allows us to treat each population independently of the others.

Nevertheless, several qualitative takeaways emerge from our modelling that we think are interesting:  for example, in our model, the persistence of inequality (multiple attracting fixed points) depends on the university using a selection rule that intentionally takes into account wealth, rather than just talent. We show that if the  university puts $\alpha = 1$ weight on type in our model, there is only a single fixed point. This suggests that changes in admission policies that reduce the focus on wealth (for example, switching to need-blind admissions and  reducing or eliminating legacy admissions) might have beneficial long-term effects. As another takeaway, we find that aggressive interventions (in our model, that aim to  lift the lower-wealth population in one shot to the basin of attraction of the higher-wealth fixed point) are often the most cost effective in the long run, compared to more modest interventions that would accomplish the same goal after $k > 1$ rounds. On the other hand, incremental interventions become optimal when society heavily discounts the future. This suggests that institutions that are able to formulate longer-term goals, but have the resources to act on them immediately (for example, non-profit universities with large endowments) may be able to combat wealth inequality more effectively.

\subsection{Related Work}
Our paper is related to economic models of inequality, which date back to \citet{arrow} and \citet{phelps}. For example, \citet{CL93} and \citet{FV92} study two stage models in which the existence of  self-confirming equilibria can cause inequality to be  persistent even when  populations are ex-ante identical.

More recently, the computer science community has begun studying dynamic models of fairness.  \citet{JJKMR17} study the costs  of imposing fairness constraints on learners in general Markov decision processes. \citet{HC18} study a dynamic model of the labor market similar to that of \citet{CL93,FV92} in which two populations are symmetric, but can choose to exert costly effort in order to improve their value to an employer. They study a two-stage model of a labor market in which interventions in a ``temporary'' labor market can lead to high-welfare symmetric equilibrium in the long run.~\citet{delayed} study a two-round model of lending in which lending decisions in the first round can change the type distribution of applicants in the 2nd round, according to a known, exogenously specified function. ~\citet{liu2020disparate} study a dynamic model where in each round, strategic individuals decide whether to invest in qualifications and the decision-maker updates his classifier that decides which individuals are qualified; they characterize the equilibria of such dynamics and develop interventions that lead to better long-term outcomes. \citet{pipelines-interventions} study a model in which decisions over individuals and populations are made along a multi-layered pipeline, where each layer corresponds to a different stage of life. They consider the algorithmic problem faced by a budgeted centralized designer who aims to intervene on the transitions between layers to obtain optimally fair outcomes, when such modifications are costly. \citet{kannan2019downstream} study a two stage model of affirmative action in which a college may set different admissions policies for an advantaged and disadvantaged group, but a downstream employer makes hiring decisions that maximize their expected objective given their posterior belief on student qualifications (that depend on the college's policies). \citet{jung2020fair} study an equilibrium model of criminal justice in which two populations with different outside option distributions make rational decisions as a function of criminal justice policy; they show that policies that have been proposed with equity considerations in mind (equalizing false positive and negative rates) actually emerge as optimal solutions to a social planner's optimization problem even without an explicit equity goal. 

We highlight two closely related papers.  \citet{heidari2021allocating} also study a model of inter-generational wealth dynamics across many rounds, in which both wealth (which they model as a binary) and talent play a role in success, as a function of opportunities that can be allocated to a limited portion of the population. Like us, \citet{heidari2021allocating} use college admissions as a running example of an institution allocating the opportunities, and like us, study a model in which admissions to college plays the role of determining wealth increase or decrease from one generation to the next. Our models differ in a number of specifics, but the primary difference between these two works is that \citet{heidari2021allocating} study the optimal policy for a very patient institution interested in maximizing its long-run payoff, and show that it recovers a form of affirmative action, preferentially offering opportunities to the less wealthy population so that it can reap the benefits of their resulting increased wealth in future generations. In contrast, we study institutions that are myopic, and optimize only for their short-term reward. In this setting, we study conditions under which such institutions do or do not perpetuate inequality, and study interventions in the settings in which they do.  \citet{mouzannar2019fair} study a continuous time dynamic in which two populations with binary, fully observable type are selected by a college with a myopic objective, whose selection rate within each population changes their type distribution at the next round. They study conditions under which imposing an ``affirmative action'' constraint (having the same selection rate within each population) can lead to equal or improved outcomes in the long run. In addition to various differences in the models (they study a setting with fully observed, binary types) and the class of interventions considered (we study changing what is observable to the university and offering direct subsidies), the most salient difference is that our model seeks to understand how wealth evolves over the long run, and has the crucial feature that wealth can be partially conflated with type in the signal observed by the university. We also focus on intervening directly on wealth via monetary subsidies. In contrast, \citet{mouzannar2019fair} do not model wealth, and in their dynamic, the type distribution directly evolves and is fully observable. Hence their myopic college does not need to do any inference as in our model. 

\section{Preliminaries}

\begin{definition}[Attracting fixed points]
Let $f:~\reals \to \reals$ be a real-valued function and let $x^*$ be such that $f(x^*) = x^*$. We call $x^*$ a fixed point of $f$. Further, let $a_t(x)$ be the sequence defined by $a_0 = x$ and $a_{t+1} = f(a_t)$; we say that $x^*$ is attracting for $x$ if and only if $a_t(x)$ converges to $x^*$.
\end{definition}

\begin{claim}[Attracting fixed points]\label{clm:attract}
Let $f$ be a real-valued, continuous, non-decreasing function such that $x^*$ is a fixed point of $f$. If $f(x) > x$ for all $x \in [a,x^*)$, $x^*$ is attracting on $[a,x^*)$. Similarly, if $f(x) < x$ for all $x \in (x^*,b]$, $x^*$ is attracting on $(x^*,b]$.
\end{claim}

\begin{proof}
Let $x \in [a,x^*)$. Let $a_t(x)$ be the sequence defined by $a_0 = x$ and $a_{t+1} = f(a_t)$.  $a_1 = f(a_0) > a_0$, Since $f$ is non-decreasing  $a_2 = f(f(a_0)) \geq f(a_0)$, but $f(x) > x$ $\forall x \in [a,x^*)$, so this inequality is in fact strict, i.e $a_2 > a_1$.
Note that by induction, we have for all $t$ that $x^* = f(x^*) \geq a_{t+1}(x) = f(a_t(x)) > a_t(x) \ldots > a_0$; hence $a_t$ is increasing and $a_t \in [a,x^*]$ for all $t$. In particular, $a_t$ is a convergent sequence with a finite limit in $[a,x^*]$. Now, since $f(x) > x$ for all $x < x^*$, $x^*$ is $f$'s unique fixed point on $[a,x^*]$. Because $f$ is continuous, we must have $\lim_{t \to +\infty} a_{t+1} = \lim_{t \to +\infty} f(a_{t}) = f \left( \lim_{t \to +\infty} a_t \right)$, i.e. the limit $l$ must satisfy $f(l) = l$. The only point on $[a,x^*]$ that satisfies this condition is $x^*$, yielding the first part of the result. A similar argument holds for the second part of the proof. 
\end{proof}

\section{Model}

We consider a university that has a non-atomic set of applicants from two  different sub-populations (or groups), denoted $1$ and $2$, and must decide which applicants to admit. Each applicant has a type $T$, where the types are random variables drawn i.i.d. from a known distribution $\cD$; we assume that the distribution of types is the same for both groups. 
Further, each applicant also has a wealth $W_i$; wealth is drawn i.i.d. from a known distribution $\cW_i$ which may depend on the applicant's group $i$. We assume that wealth and types are drawn independently of each other.

The university aims to make decisions based on each applicant's type and wealth, and would like to admit individuals for whom 
\[
\alpha T + (1 - \alpha) W \geq \thr,
\]
for some parameter $\alpha \in [0,1]$ and some threshold $\thr$. Here, $\alpha$ represents how much the university values type versus wealth. However, the university cannot directly observe this quantity. Instead, we assume that it can only see a score $S$ for each applicant (for example in the form of a standardized test result), where the score is a convex combination of both an applicant's type $T$ and wealth $W$. We write 
\[
S = \beta T + (1 - \beta) W,
\]
for some known $\beta \in [0,1]$. 
The university then performs a Bayesian update and admits individuals that satisfy
\[
\E_{T,W} \left[\alpha T + (1- \alpha) W | S \right] \geq \thr. 
\]

We are interested in understanding the long-term dynamics of a process where the university's decisions (made as described above) affect the individuals' future attributes. We consider a discrete time horizon, in which at each time step $t \in \mathbb{Z}^+$, the university's decisions shapes the distribution of wealth in each group in time step $t + 1$. In particular, we assume that the expected wealth $\mu_i^{t+1}$ of group $i$ in step $t+1$ is the fraction of group $i$ that is admitted by the university at time step $t$. I.e., we write 
\begin{align}\label{eq: update_rule_def}
\mu_i^{t+1} = \Pr_S \left[\E_{T,W} \left[\alpha T + (1- \alpha) W | S \right] \geq \thr \right]
\end{align}

This is motivated by the fact that students that are admitted to competitive universities are expected to reach better life outcomes and accumulate more wealth. 

In the rest of the paper, we make the following assumptions on the functional form of the type and wealth distributions:
\begin{assumption}
The  type $T$ at any time instant is drawn from the distribution $T \sim \cN\left(0,\stdt^2\right)$. The initial wealth at time $0$ for population $i$  satisfies $\mu^0 \in [0,1]$, and $W_i \sim \cN \left(\mu_i^t,\stdw^2\right)$ at time step $t$ for a fixed constant $\stdw$.
\end{assumption}
Note that the type can be centered around $0$ without loss of generality, by changing the value of $\thr$ used by the employer by the corresponding amount. The assumption $\mu^0 \in [0,1]$ is also without loss of generality, and simply renormalizes the average wealth of a group to be between $[0,1]$, so long as we consider populations with bounded wealth.
Note that with our assumptions the mean wealth of each group always stays in the range $[0,1]$
although the sampled wealth of individuals can fall outside this interval. 


\section{Wealth Dynamics and Properties}

We note that the dynamics of each group only depend on the decisions made by the university within that group. Therefore, we can treat groups independently. In this section, we focus on a single group at a time, and drop the dependencies on $i$ in our notations for simplicity. We show that several attracting fixed points can arise from our dynamics; in particular, there are regimes of parameters under which there is a low wealth fixed point that groups with initially low wealth converge to, and a high wealth fixed point that groups with initially high wealth converge to. In Section~\ref{sec: interv_subsidies}, we consider interventions that apply to more general update functions that the ones described in this Section, so long as they have similar fixed point properties. 

\subsection{Computing the Wealth Update Rule} We start by characterizing the joint distributions of the type $T$, the wealth $W$, and the score $S$. 

\begin{claim}\label{clm:bivar}
Let $\mu \triangleq \E [W]$. We have that $(T, S)$ forms a bivariate Gaussian distribution with mean $(0, (1-\beta)\mu)$ and covariance matrix
$$
\begin{bmatrix}
\stdt^2 &  \beta \stdt^2\\
\beta \stdt^2 & \beta^2\stdt^2 + (1-\beta)^2\stdw^2.
\end{bmatrix}
$$
Similarly, $(W, S)$ forms a bivariate Gaussian distribution with mean $(\mu, (1-\beta)\mu)$ and covariance matrix
$$
\begin{bmatrix}
\stdw^2 &  (1-\beta) \stdw^2 \\
(1-\beta) \stdw^2 & \beta^2\stdt^2 + (1-\beta)^2\stdw^2.
\end{bmatrix}
$$ 
\end{claim}

The proof is provided in Appendix~\ref{app:bivar}. This allows us to compute the update function that maps the wealth of a group in the current round, $\mu^t$, to the wealth of that same group in the next round, $\mu^{t+1}$:

\begin{lemma}\label{lem:update_rule}
Recall that $\tau$ is the threshold used by the university to decide admission. At every time step $t$, we have
\begin{align*}
\mu^{t+1}  = 1 - \Phi\left( K\left(\alpha, \beta, \stdt, \stdw\right)\left(\thr-(1-\alpha) \mu^t\right) \right),
\end{align*}
where $K\left(\alpha, \beta, \stdt, \stdw)\right)  \triangleq \frac{\sqrt{\beta^2 \stdt^2 + (1 - \beta)^2 \stdw^2}}{\alpha \beta \stdt^2 + (1 - \alpha) (1 - \beta) \stdw^2}$ and $\Phi$ is the cumulative density function of a standard Gaussian. We denote the update rule function
\begin{align}\label{eq: update_rule}
f(x) \triangleq 1 - \Phi\left(K\left(\alpha, \beta, \stdt, \stdw\right) \left(\thr-(1-\alpha) x\right) \right).
\end{align}
\end{lemma} 

For simplicity of notation, we omit the dependency of $f$ in the parameters of the problem when clear from context. When not, we explicitly write the dependency of $f$ in the parameters of interest. The proof of Lemma~\ref{lem:update_rule} is mostly algebraic, and is provided in Appendix \ref{app: proof_update_rule}.

\subsection{Fixed Points and Convergence of the Dynamics}
We can now use the closed-form expression for the update rule to study the properties of the wealth dynamics. In this section, we bound the number of fixed points of our dynamics, provide properties of these fixed points, and characterize which fixed point each initial wealth converges to. We start by noting that the update rule has a simple shape. Indeed:
\begin{claim}\label{clm: f_shape}
$f(x)$ is continuous and increasing in $x$. Further, $f$ is convex on $[0,x^*]$ and concave on $[x^*,1]$ where 
\begin{align*}
x^* 
= \begin{cases}
0 &\text{if}~\thr \leq 0,
\\\frac{\thr}{1 - \alpha} &\text{if}~0 < \thr < 1 - \alpha,
\\1~\text{if} &\thr \geq 1 - \alpha.
\end{cases}
\end{align*}
\end{claim}

The proof of the above claim is given in Appendix~\ref{app: f_shape}. We now use the above properties on the shape of $f$ to derive properties of its fixed point. First we remark that $f$ has at least one fixed point, since $f(0) > 0$ and $f(1) < 1$, and $f$ is continuous. Now, note that the number of fixed points of $f$ is also upper-bounded:
\begin{lemma}\label{lem:fixed_pts}
Suppose $0 < \thr < 1 - \alpha$, then $f(x) = x$ has at most $3$ solutions for $x \in [0,1]$. If $f$ has $3$ fixed points $z_1 < z_2 < z_3$, it must be that $z_1 < \frac{\tau}{1-\alpha} < z_3$. If $\thr \leq 0$ or $\thr \geq 1-\alpha$, $f(x) = x$ only has a single solution for $x \in [0,1]$.
\end{lemma}

The proof is provided in Appendix~\ref{app:fixed_pts}. Lemma~\ref{lem:fixed_pts} has direct implications for disparities across groups with different starting expected wealth. In particular, the number of fixed points of $f$ determines whether different groups must converge to equal wealth (the case in which there is only a single fixed point) in the long-run or whether there are cases in which wealth inequality is persistent (the case in which there are multiple fixed points). We discuss these implications in more details in the rest of this section. 

\paragraph{The Case of a Single Fixed Point} We now study conditions under which $f$ has single vs. multiple fixed points. We first consider the case of a single fixed point. In this case, we remark that the single fixed point has the following property:
\begin{claim}\label{clm: oneFP_attracting}
If $z$ is the single fixed point of $f$, then $z$ is attracting on $[0,1]$.
\end{claim}

\begin{proof}
Since $f(0) > 0$, $f(1) < 1$, and $f$ is continuous and has a single fixed point $z$, it must be that $f(x) > f(z) = z$ for $x < z$ and $f(x) < f(z) = z$ for $x > z$. \ar{This could use a line of justification.} Applying Claim~\ref{clm:attract} concludes the proof. 
\end{proof}

This implies in particular that when $f$ has a single fixed point $z$, wealth dynamics converge to this fixed point no matter what the starting wealth was. This means in particular that there are no long-term disparities between populations of different initial socio-economic statuses (though they may take different amounts of time to reach the same wealth), i.e. the dynamics self correct for initial wealth disparities. Figure~\ref{fig: 1 FP with cobweb} shows an instantiation of a wealth update functions with a single fixed point and the corresponding wealth dynamics (in green); the plots illustrate convergence of the dynamics to the single fixed point starting both from an initially low wealth  (Figure~\ref{fig: 1 FP with cobweb} (a)) and from an initially high wealth (Figure~\ref{fig: 1 FP with cobweb} (b)).

\begin{figure}%
    \centering
    \subfloat[\centering Low initial wealth]{\includegraphics[width=6cm]{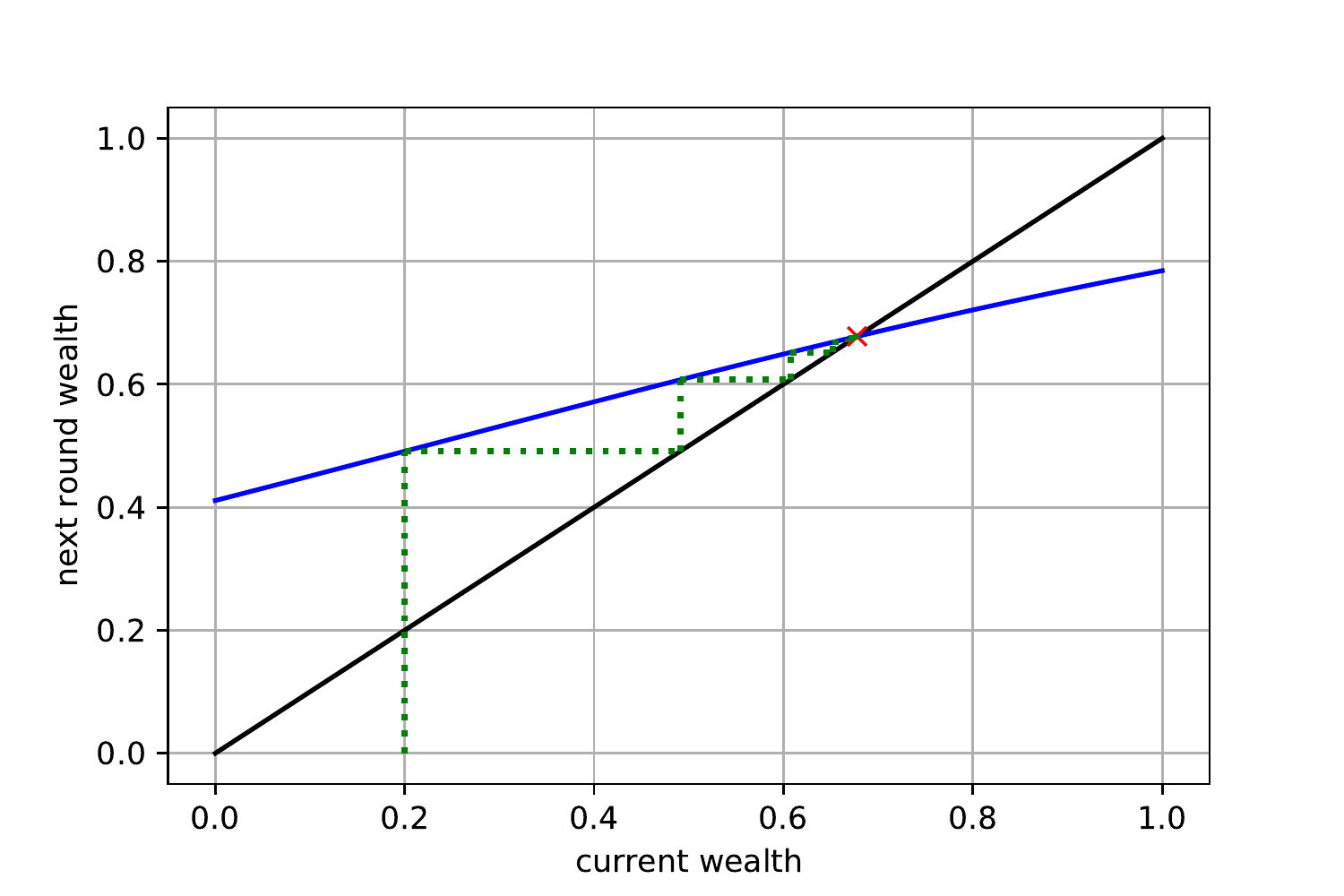}}%
    \qquad
    \subfloat[\centering High initial wealth]{{\includegraphics[width=6cm]{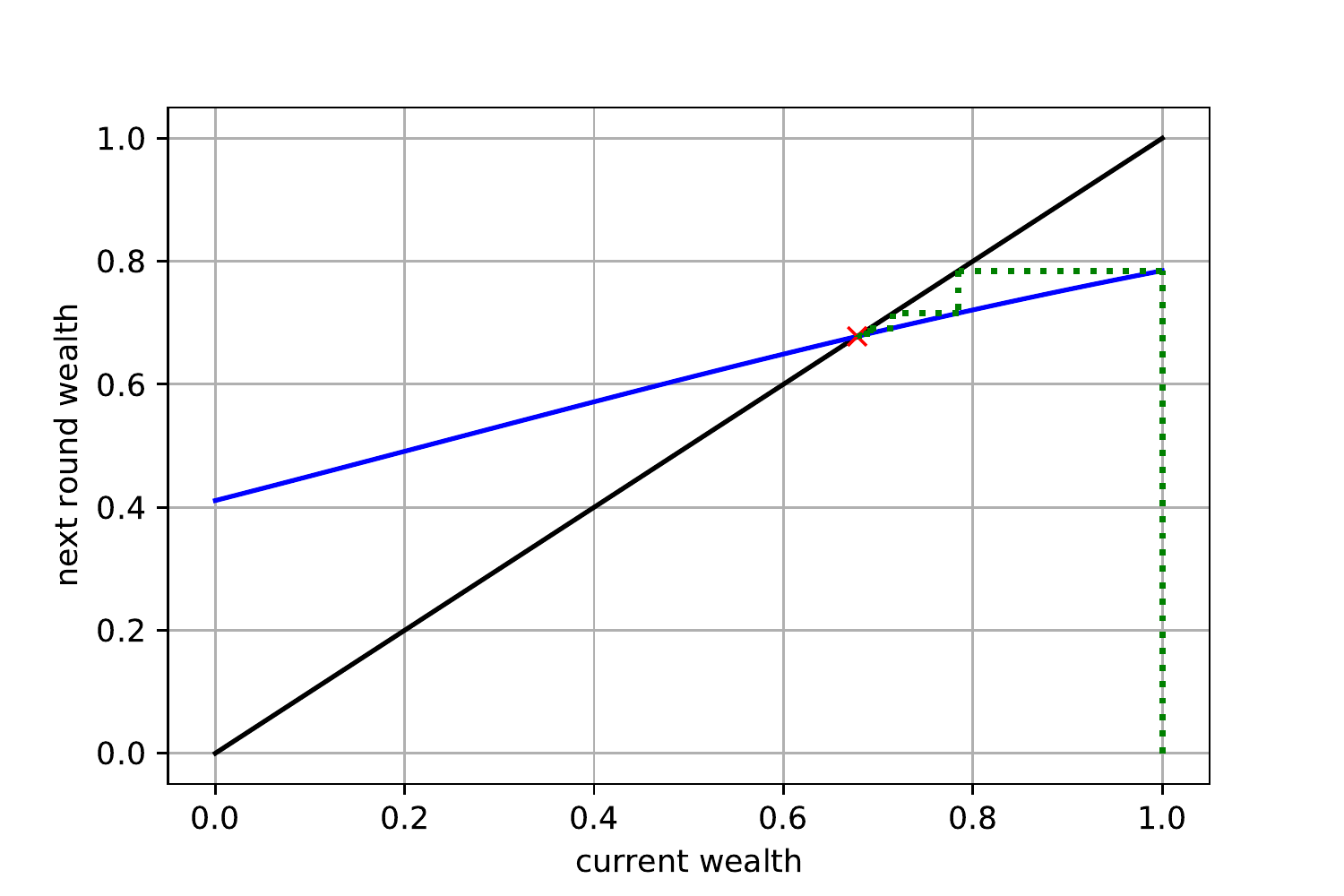}}}%
    \caption{A wealth update function with a single fixed point, for $\alpha = 0.1, \beta = 0.6, \gamma = 0.4, \sigma =  1.1, \tau =0.2$. The update function is plotted in blue, its single fixed point in red, and the wealth dynamics induced by the update function in green. Sub-figure (a) considers dynamics starting at an initial wealth of $0.2$ while sub-figure (b) considers dynamics starting at wealth $1.0$.}
    \label{fig: 1 FP with cobweb}%
\end{figure}

We note that Lemma~\ref{lem:fixed_pts} already implies that there exist interesting situations in which the dynamics have a single attracting fixed point and wealth dynamics are self-correcting. The first one is when $\thr$ is small ($\thr < 0$); i.e., the university is not very selective in its admissions. Intuitively, this leads to most individuals from any group being admitted (almost) independently of their starting wealth, which allows even economically disadvantaged groups to build wealth over time. The other situation deriving from Lemma~\ref{lem:fixed_pts} arises when $\thr > 1 - \alpha$. This can arise for two reasons: first, is the university is very selective and sets high values of $\thr$, wealth becomes insufficient to qualify an individual for admission (as then $(1 - \alpha)E[W|S] \leq 1 - \alpha < \thr$); an agent must have sufficiently high (inferred) type to be admitted, which helps reduce disparities due to wealth. This can also arise when $\alpha$ is large and the university is mostly interested in type over wealth. Intuitively, in this case, the university pays significant attention to their posterior belief on the type of an individual, which facilitates equalizing the treatment of groups of different wealth since they have the same type distributions; while the university cannot observe type directly, they discount for average wealth more (by a factor of $(1 - \alpha) \mu$) hence correct for wealth disparities more as $\alpha$ is smaller. 

Below, we provide an additional condition under which $f$ has a single fixed point:
\begin{claim}\label{clm: single_fp_condition}
If $K(\alpha,\beta,\stdt,\stdw) \leq \frac{\sqrt{2 \pi}}{1 - \alpha}$, $f$ is a contraction mapping and has a unique attracting fixed point.
\end{claim}

\begin{proof}
This immediately follows from $f'(x) = \frac{K (1-\alpha)}{\sqrt{2 \pi}} \exp \left(- K^2 (\thr - (1-\alpha) x)^2/2\right)$ and from 
\\$\exp \left(- K^2 (\thr - (1-\alpha) x)^2/2\right) \leq 1$ (with equality at $x = \thr/(1-\alpha)$).  Note that $f(0) > 0$ and $f(1) < 1$ so the fixed point $z$ must satisfy $f(x) < z$ if and only if $x < z$ and must be attracting.
\end{proof}

We note that $K(\alpha,\beta,\stdt,\stdw) = \frac{\sqrt{Var(S)}}{Cov(D, S)}$ where $D = \alpha T + (1-\alpha) W$ \kri{Can we use some other letter for university objective, since $\cD$ was used for distribution}. This implies that, holding the college's objective function (i.e., $\alpha$) constant, the better the scoring rule aligns with the university's admissions criteria (i.e. as the covariance between $D$ and $S$ increases), the smaller $K$ becomes. This makes the condition  that $f$ is a contraction mapping with a single fixed point easier to satisfy, which in turn causes wealth dynamics to self-correct for initial inequality. \kri{But $\alpha \to 1$, always gives contraction, since condition always satisfied}\juba{good catch, fixing} When $\alpha \to 1$, the condition is always satisfied, and $f$ has a single fixed point. This may not be surprising in that in this case the university only cares about type in admissions, and the university requires a higher threshold on scores for wealthier populations; this helps reduce disparities across populations with disparate wealth.

\paragraph{The Case of Multiple Fixed Points} We first characterize which fixed points are attracting when multiple points arise, and which regime of initial wealth lead to which fixed points. We focus on the case of three fixed points, as the case of two fixed points is a corner case than can only arise if $f(x)$ is tangent to $Id(x) = x$ at one of the fixed points.\footnote{Suppose this is not the case. $f(0) > 0$ hence $f(x) > x$ before the first fixed point. Because it is not tangent to the identity line, it must then be that $f(x) < x$ between the first and the second fixed point. Similarly, it must then be that $f(x) > x$ after the second, last fixed point. This contradicts $f(1) < 1$.} 

\begin{claim}
Suppose $f$ has $3$ fixed points, denoted $z_1 < z_2 < z_3$. Then $z_1$ is attracting for $[0,z_2)$ and $z_3$ is attracting for $(z_2,1]$. 
\end{claim}

\begin{proof}
This follows from the proof of lemma~\ref{lem:fixed_pts}. Indeed, let $g(x) = f(x) - x$, we have that $g(0) = f(0) > 0$, then $g$ must decrease below $0$, increase above $0$, and decreases below $0$ again as $g(1) = f(1) - 1< 0$. This implies that $f(x) > x$ for $x < z_1$ and $x \in (z_2,z_3)$, while $f(x) < x$ for $x \in (z_1,z_2)$ and $x > z_3$. 
\end{proof}

In particular, when there are three fixed points, a population that starts with low wealth will converge to the first fixed point, while a group with large initial wealth will converge to the third fixed point. In this case, initial disparities in wealth persist in the long term, and interventions are needed for different populations to obtain equitable long-term wealth outcomes. Figure~\ref{fig: 3 FP with cobweb} shows a wealth update function with three fixed points and the corresponding dynamics for two different starting points. We note that starting at low wealth leads to convergence to the first and lowest fixed point, while starting at relatively high health leads to convergence to the highest fixed point. The figure illustrates how wealth disparities can propagate and amplify over time.  

\begin{figure}%
    \centering
    \subfloat[\centering Low initial wealth]{\includegraphics[width=6cm]{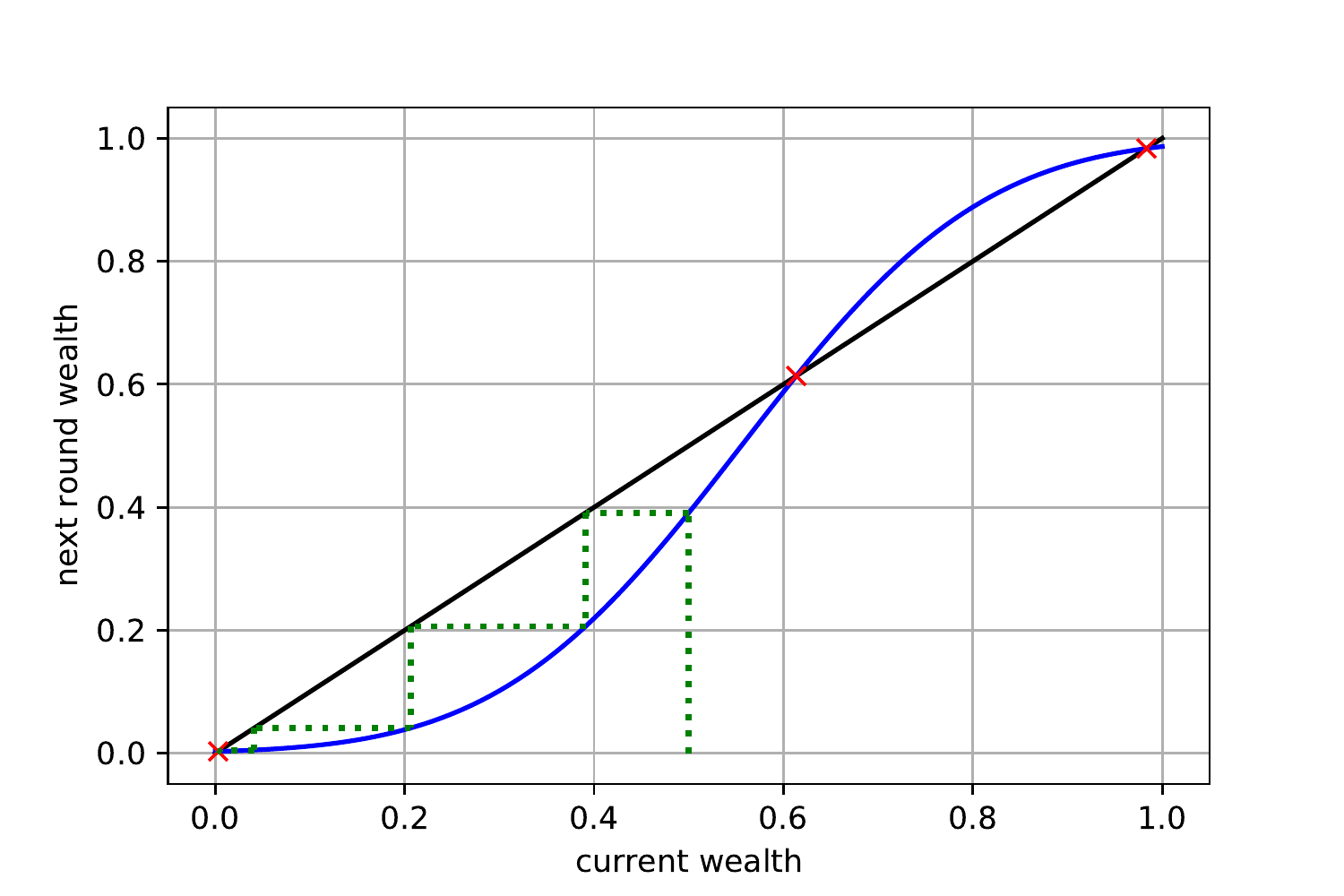}}%
    \qquad
    \subfloat[\centering High initial wealth]{{\includegraphics[width=6cm]{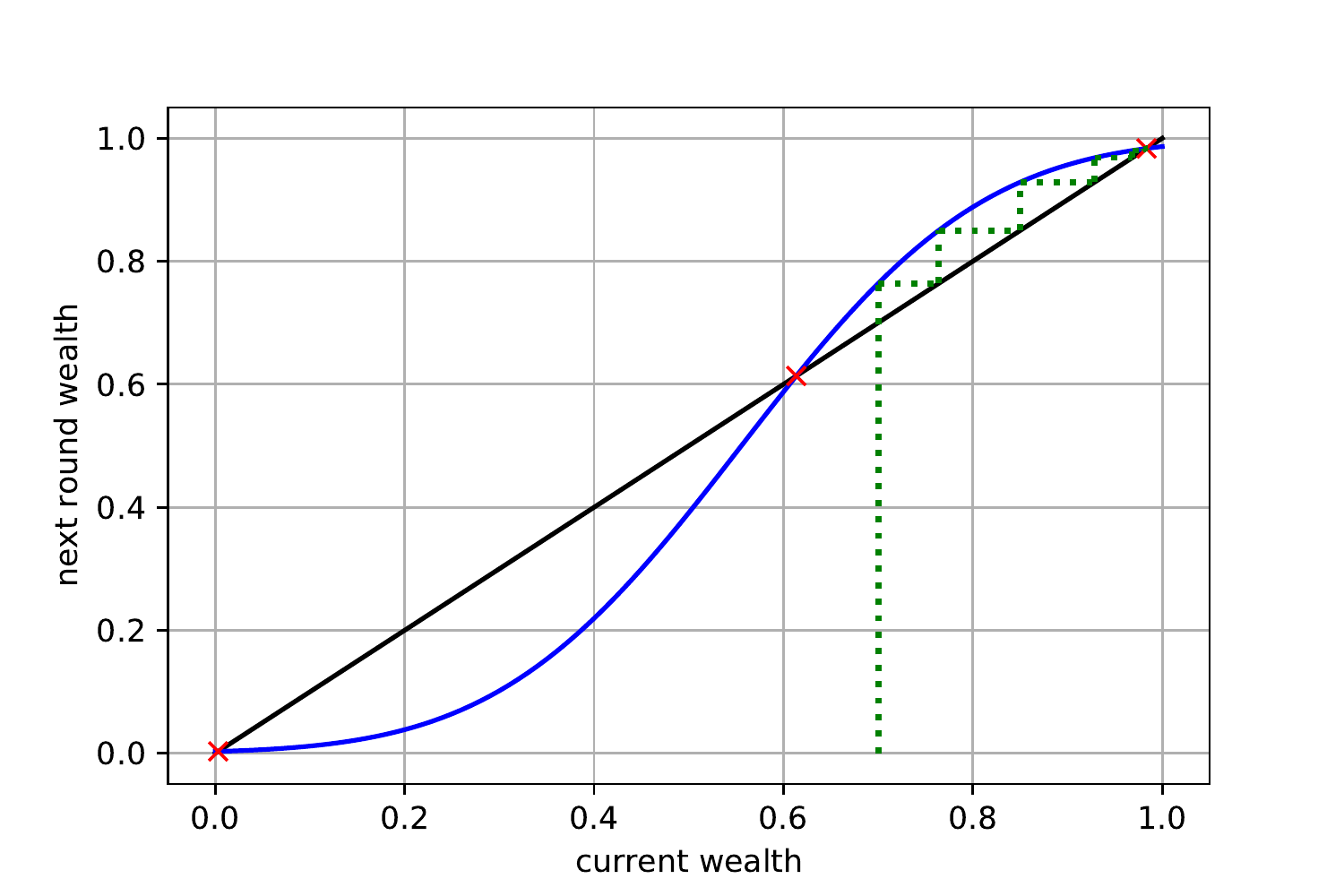} }}%
    \caption{A wealth update function with 3 fixed points, for $\alpha = 0.1, \beta = 0.95, \gamma = 1.4, \sigma =  1.1, \tau = 0.5$. the update function is plotted in blue, its fixed points in red, and the wealth dynamics induced by the update function in green. Sub-figure (a) considers dynamics starting at an initial wealth of $0.5$ while sub-figure (b) considers dynamics starting at an initial wealth of $0.7$.}
    \label{fig: 3 FP with cobweb}%
\end{figure}

We note that such situations can only arise in the regime in which $0 < \thr < (1-\alpha)$. In particular: 
\begin{claim}\label{clm:condition_3fp}
Suppose $K(\alpha,\beta,\stdt,\stdw) > \frac{\sqrt{2 \pi}}{1 - \alpha}$ and $\thr = \frac{1-\alpha}{2}$. Then $f$ has $3$ fixed points.
\end{claim}

\begin{proof}
In this case, note that $x^* = \frac{\thr}{1-\alpha} = \frac{1}{2}$. Further, we know that 
\begin{align*}
f(x^*) 
&= 1 - \Phi\left(K\left(\alpha,\beta,\stdt,\stdw\right) \cdot \left(\thr - (1 - \alpha) x^*\right)\right) 
\\&=  1 - \Phi(0) 
\\&= 1/2,
\end{align*}
implying that $x^* = 1/2$ is a fixed point of $f$. Further, 
\begin{align*}
f'(x^*) = \frac{K (1 - \alpha)}{\sqrt{2 \pi}} > 1,
\end{align*}
hence $f(x) < x$ in a small neighborhood $(x^*-\varepsilon,x^*)$ and $f(x) > x$ in a small neighborhood $(x^*,x^*+\varepsilon)$. By continuity of $x$ and the fact that $f(0) > 0$, $f(x) = x$ must have a solution on $[0,x^*)$. Similarly, since $f(1) < 1$, $f(x) = x$ must have a solution on $(x^*,1]$.
\end{proof}

Note that because $f$ is continuous in $\thr$, $f$ must have three fixed points for any $\thr$ in a neighborhood of $ \frac{1-\alpha}{2}$. I.e., there exists a continuous range of values of $\thr$ for which $f$ has three fixed points, showing that such situations are not a corner case of our framework, unlike when $f$ has two fixed points. 

\begin{remark}
There is a gap between the conditions given in Claim~\ref{clm: single_fp_condition} and Lemma~\ref{lem:fixed_pts} under which a single fixed point arises, and the condition given in Claim~\ref{clm:condition_3fp}. In particular, we remark that even if $f$ is not a contraction mapping and $K > \frac{\sqrt{2 \pi}}{1 - \alpha}$, or when $0 \leq \tau \leq 1 - \alpha$, it may still have only a single fixed point. To investigate how often $3$ fixed points can arise, we picked a uniform grid of parameter values $(\alpha,\beta,\gamma,\sigma,\thr) \in [0,1]^5$ and investigated what fraction of the parameters that satisfy either $0 \leq \tau \leq 1 - \alpha$ or $K > \frac{\sqrt{2 \pi}}{1 - \alpha}$ actually lead to an update rule with three fixed points. We found that this was the case for roughly $30$ percent of the values we explored, implying the existence of a significant range of parameters for which there are disparities in the long term wealth of different populations. 
\end{remark}

\label{sec:wealth_dynamics}

\section{Interventions To Improve Long Term Population Wealth}\label{sec:interventions}

In this section, we consider different types of intervention aiming at improving and equalizing population wealth, when the wealth dynamics have multiple fixed points (and so are not necessarily self correcting). We consider three types of interventions: i) changing the design of the admission rule used by the university, ii) changing the design of the standardized test or scoring rule that the university relies on, and iii) providing subsidies to disadvantaged groups. 

\subsection{Changing the Admission Rule}

Since in our model, it is admission to university that confers a wealth advantage to the next generation, a natural intervention is to increase the capacity of the university, thereby admitting more people. Rather than changing the objective function of the university $(\alpha)$, we model this kind of intervention by decreasing the university's admissions threshold $\tau$. We note that although we do not explicitly quantify it, increasing the capacity of a university will come at some financial cost, and so this kind of intervention is not necessarily incomparable to the direct subsidies we consider later.

The claim below characterizes how the fixed points of $f$ change when we change the value of $\thr$. For the sake of notation, we let $f(.,\thr)$ be the update rule when the chosen threshold is $\thr$, and omit the dependencies of $f$ in the other parameters of the problem.

\begin{theorem}
\label{thm:threshold}
Fix $\alpha,\beta,\stdt,\stdw$. For a given admission threshold $\thr$, let $z_1(\thr) < z_2(\thr) < z_3(\thr)$ be the fixed points of $f$ when all $3$ exist. Let $\thr' < \thr$ be such that $f(.,\thr')$ has $3$ fixed points, we have
\[
z_1(\thr') > z_1(\thr)~~\text{and}~~ z_3(\thr') > z_3(\thr), 
\]
but
\[
z_2(\thr') < z_2(\thr).
\]
\end{theorem}

\begin{proof}[Proof Sketch]
The proof follows simply by showing that the update function is decreasing in $\thr$. In turn, decreasing $\thr$ moves the update rule $f$ ``up'', which increases attracting fixed points and decreases unstable ones. The full proof is given in Appendix~\ref{app:threshold}.
\end{proof}

In interpreting Theorem \ref{thm:threshold}, we recall that only the first and third fixed points $z_1$ and $z_3$ are attracting, and that $z_2$ is unstable (has no points $x$ for which it is attracting for $x \neq z_2$). Hence, if we are in a situation in which there is persistent wealth inequality (multiple fixed points), we find that if we can decrease the admissions threshold $\tau$ of the university, then \emph{either}:
\begin{enumerate}
\item We increase the wealth of both of the attracting fixed points (and hence the wealth of both populations, irrespective of which attractive fixed point they are at). By decreasing $z_2$ we also reduce the size of the attracting region $[0, z_2)$ of the lower wealth fixed point, thus enabling poorer populations to converge to the most desirable fixed point. 
Or 
\item We move the dynamics to one in which there is only a single fixed point, and hence eliminate wealth inequality.
\end{enumerate}



\subsection{Changing the design of the scoring rule $S$}

The college has to engage in inference about an applicant's type and wealth when $\beta \neq \alpha$, because the signal it receives does not align with its objective function. What if we can modify the signal (by e.g. changing the design of a standardized test) to more closely align the signal with the college's objective? 

In this section we characterize how the fixed points of $f$ change when we change the value of $\beta$. We denote $f(.,\beta)$ the update rule when the scoring rule uses parameter $\beta$, while the other parameters of the problem remain fixed.

\begin{theorem}
\label{thm:beta}
Fix $\alpha,\tau,\stdt,\stdw$. For a given $\beta$, let $z_1(\beta) < z_2(\beta) < z_3(\beta)$ be the fixed points of $f$ when all $3$ exist. Suppose $\beta < \alpha$ and let $\beta' \in (\beta,\alpha)$ be such that $f(.,\beta')$ has $3$ fixed points, then 
\[
z_1(\beta') > z_1(\beta)~~\text{and}~~ z_3(\beta') < z_3(\beta).
\]
Similarly, if $\beta > \alpha$ and $\beta' \in (\alpha,\beta)$, we have 
\[
z_1(\beta') > z_1(\beta)~~\text{and}~~ z_3(\beta') < z_3(\beta).
\]
\end{theorem}

\begin{proof}[Proof Sketch]
The first part of the proof follows simply by showing that the update function is increasing in $\beta$ for $x < \tau/(1-\alpha)$ (where the first fixed point lies) and decreasing in $\beta$ as for $x > \tau/(1-\alpha)$ (where the third fixed point lies). In turn, increasing $\beta < \alpha$ towards $\alpha$ move the update rule $f$ ``up'' around the first fixed point and ``down'' around the third fixed point, which increases $z_1$ and decreases $z_3$. A similar argument holds for $\beta > \alpha$. The full proof is given in Appendix~\ref{app:beta}. 
\end{proof}

Intuitively, one might suppose that to reduce wealth disparities, we should redesign tests so as to make them reflect type more strongly and wealth less strongly (since types are distributed identically across groups). But Theorem \ref{thm:beta} shows that counter-intuitively, this need not be the case\footnote{Even when $\beta \to 1$, $f$ may have three fixed points: by Claim~\ref{clm:condition_3fp}, this arises for example when $K(\alpha,1,\stdt,\stdw) = \frac{1}{\alpha \stdt} > \frac{1 - \alpha}{\sqrt{2 \pi}}$ and $\tau = \frac{1 - \alpha}{2}$. In this case, setting $\beta = \alpha$ surprisingly leads to better outcomes than $\beta = 1$.}. Instead, what Theorem \ref{thm:beta} shows is that in order to reduce inequality, we want to move $\beta$ towards $\alpha$, causing the signal to better reflect the objective function of the college ---- even when this results in reducing the extent to which the signal reflects type\footnote{Of course, if we can, we would prefer to increase the extent to which the college \emph{values} type rather than wealth, but to the extent that we cannot do this, then we want to align the test with the college's objective.}. Theorem \ref{thm:beta} shows that moving $\beta$ towards $\alpha$ always has the effect of reducing wealth disparities. It either:
\begin{enumerate}
    \item Increases the wealth of the less wealthy attracting fixed point, and decreases the wealth of the more wealthy attracting fixed point, thereby decreasing the long term wealth disparity, or it
    \item Shifts the dynamic to one that has only a single fixed point, thereby eliminating long term wealth disparities. 
\end{enumerate}


\subsection{Direct Subsidies}\label{sec: interv_subsidies}
We have thus far considered interventions that can be applied by the college (admitting more students) or a testing body (changing the design of the signal). In this section, we take the point of view of a funding body or governmental agency that can provide direct monetary subsidies to populations. We generalize the class of functions we study to include any function $f$ satisfying the following properties:

\begin{assumption}
$f$ is continuous and increasing. Further, $f$ has three fixed points $z_1 < z_2 < z_3$ with $f(x) > x$ on $[0,z_1]$ and $[z_2,z_3]$ and $f(x) < x$ on $[z_1,z_2]$ and $[z_3,1]$. 
\end{assumption}

The above assumption captures the main properties of our function $f = 1 - \Phi\left(K\left(\alpha, \beta, \stdt, \stdw\right) \left(\thr-(1-\alpha) x\right) \right)$ when it has three fixed points and implies the same attracting properties we established for $z_1$, $z_2$, and $z_3$, but also encompasses more general update rules that need not result from the Gaussian inference process we have studied thus far. We note that this allows us to study general S-shaped function with diminishing returns at both ends of the socio-economic spectrum. Such functions model situation in which people of very low income or very high income see little upward mobility (in the first case because of a lack of access to opportunities, and in the second case due to the fact that individuals of higher income are rare), whereas middle income individuals have significant opportunities to improve their wealth.

We denote by $C(\mu,t)$ the subsidy given to a population with wealth $\mu$ at time step $t$. The wealth of a population $t+1$ then depends of the wealth in time $t$ as
\[
\mu^{t+1} = f\left(\mu^t + C(\mu^t,t)\right).
\]
In this setting, we consider interventions that allow a population to reach beyond the second fixed point $z_2$. Once a population reaches wealth (even slightly) over $z_2$, their wealth naturally evolves to the highest attracting fixed point $z_3$ over time; i.e., wealth dynamics self-correct for disparities with no intervention needed. For the same reason, we only consider $\mu^0 \in [z_1,z_2]$; this is because populations with $\mu^0 < z_1$ will converge to $z_1$ without intervention, and we can start intervening once $\mu^0$ reaches $z_1$, while a population with $\mu^0 > z_2$ will reach the best long-term outcome (the highest fixed point, $z_3$) on its own. Therefore, from now on, we assume $C(\mu) = 0$ for all $\mu \notin [z_1,z_2]$. We can now formulate our centralized designer's objective, which is to minimize the following loss function: 
\[
L(C) = \lambda \sum_{t = 0}^{T(C) - 1} \discnt^t C(\mu^t,t) + (1 - \lambda) \sum_{t = 0}^{T(C) - 1} \discnt^{t} (z_2 - \mu^{t}),
\]
where $\discnt, \lambda \in [0,1)$, and $T(C) = \min \{t~~\text{s.t.}~\mu^t \geq z_2\}$ is the first time step such that $\mu^t \geq z_2$. Here $\discnt$ is a discounting factor; the lower $\discnt$ is, the less the designer cares about future as opposed to immediate outcomes. The objective is a convex combination of two terms, with weights controlled by $\lambda$. The first term consists of the discounted monetary cost of the subsidies (the sum goes up to time $T(C)-1$, since after the wealth of the population crosses $z_2$, the subsidies cease. This term represents a preference to spend less money on direct subsidies. The second term consists of the sum discounted difference between the target wealth $z_2$ that the intervention is aiming at, and the wealth of the population at the current round. This term represents a preference to quickly increase the wealth of the lower wealth population. $\lambda$ represents the relative strength of these two preferences.

Note that $z_2 - \mu^0$ is a constant term that does not depend on the designer's interventions, hence we will equivalently aim to minimize 
\[
L(C) = \lambda \sum_{t = 0}^{T(C) - 1} \discnt^t C(\mu^t,t) + (1 - \lambda) \sum_{t = 1}^{T(C) - 1} \discnt^{t} (z_2 - \mu^{t}),
\]
where we drop the discounted difference between $z_2$ and the initial wealth $\mu^0$ at $t = 0$.


\paragraph{Algorithmically finding a near-optimal subsidy function $C(.)$} We note that in our setting, one may discretize the space of possible costs and use dynamic programming to find optimal interventions from each possible starting point. However, doing so requires carefully understanding the wealth update function $f$. In practice, detailed knowledge of $f$ will be hard to come by. For this reason, in the rest of this section, we will aim for a ``detail free'' solution and consider a simple class of constant subsidies and study how they can be applied with minimal information about the wealth update $f$.

\paragraph{Constant Subsidies} In the rest of this section, we consider the case in which $C(\mu)$ is constant in $\mu$ for $z_1 \leq \mu \leq z_2$. I.e. there exists $C \in [0,1]$ such that $C(\mu) = C$ for all $\mu \in [z_1,z_2]$ and $C(\mu) = 0$ otherwise.  We call these $C$-subsidy interventions. We qualify a $C$-subsidy intervention as a $k$-shot intervention if it takes $k$ time steps under the subsidy to reach wealth (at least) $z_2$ when starting at wealth $z_1$, i.e. if $T(C) = k$. \juba{changed} Note that different values of $C$ may lead to the same number of steps $k$ such that $\mu^{k} \geq z_2$, i.e. there may be several values of $C$ that qualify as a $k$-shot intervention for a given value of $k$. 
\juba{please make sure I deleted all references to ``infinitesimal'' (it's not always infinitesimal) in the below.}

Our aim is to give guidelines on how to choose $C$ while using minimal information about the function $f$. Here, we will encode this minimal information as a single, real parameter $\Delta$, defined as 
\begin{align}\label{Delta_def}
\Delta = \max_{x \in [z_1,z_2]}~~x - f(x).
\end{align}
Intuitively, $\Delta$ measures how difficult it is for a subsidy to have an effect on wealth that propogates in the next round. When $\Delta \to 0$, we have that $f(x) \to x$ on $x \in [z_1,z_2]$, and investing a subsidy of $C$ increases the population wealth by $C$, since $f(\mu^t + C) \to \mu^t + C$. However, we have that for at least one value of $\mu^t$, $f(\mu^t + C) = \mu^t + C - \Delta$, implying that when $\Delta$ is large, a large amount of the subsidy is lost in the next round, and so its overall effect is small. If we want to guarantee that our subsidies will eventually lift the lower wealth population to the higher wealth fixed point independently of its starting point, we need to consider subsidies in which $C > \Delta$. 

\begin{claim}\label{clm:min_intervention}
Suppose $C \leq \Delta$. Then there exists a starting wealth $\mu^0 \in [z_1,z_2)$ such that $\mu^t < z_2$ for all $t$; i.e., $\mu_t$ never reaches $z_2$. On the other hand, if $C > \Delta$, there exists $t$ such that $\mu^t \geq z_2$.
\end{claim}

\begin{proof}
\juba{check carefully, there are a few subtleties.}
Let $x_\Delta \in (z_1,z_2)$ be any value of $x$ such that $f(x) = x - \Delta$ (note that $x_\Delta \neq z_1,~z_2$ where $f(x) - x = 0$, since $f(x) < x$ on $[z_1,z_2]$ if $f$ has three fixed points). Suppose $\mu^t < x_\Delta - \Delta$ and $C \leq \Delta$, then $\mu^{t+1} = f \left(\mu_t + C\right) < f \left(x_\Delta - \Delta + C\right) \leq f(x_\Delta) = x_\Delta - \Delta$. I.e., $\mu_t < x_\Delta - \Delta < z_2$ for all $t$ so long as $\mu^0 \in [z_1,x_\Delta - \Delta)$; note that the interval is not empty as $x_\Delta - \Delta = f(x_\Delta) > z_1$. For the second part of the proof, note that by definition of $\Delta$, for all $t$, $\mu^{t+1} = f\left(\mu^t + C\right) \geq \mu^t + C - \Delta$, hence the group wealth increases by at least a constant amount $C - \Delta$ at each time step. 
\end{proof}

Intuitively, this holds because if $C$ is smaller than $\Delta$, it becomes insufficient to compensate the fact that the wealth of a group can decrease by an amount up to $\Delta$ at each round. In the rest of this section, we aim to understand how different interventions for different values of $C$ compare to each other, and when to choose low-cost versus high-cost interventions. Before doing so, we note that there is always a single, optimal $1$-shot intervention among all such $1$-shot interventions:

\begin{fact}
The $1$-shot intervention with cost $C = z_2 - \mu^0$ has smaller cost  than any other $1$-shot intervention. This immediately follows from the fact that any $1$-shot intervention with cost $C$ has loss $\lambda C$, and that no intervention with $C < z_2 - \mu^0$ can reach $z_2$ in one shot, as $\mu^1 = f(\mu^0 + C) < f(z_2) = z_2$. 
\end{fact}

We now provide a sufficient condition under which the $1$-shot intervention is guaranteed to be optimal.

\begin{theorem}\label{thm: 1-shot_best}
Suppose $\discnt \geq \lambda$. Then, any $k$-shot intervention has higher loss than the $1$-shot, $\left(z_2 - \mu^0\right)$-subsidy intervention. I.e. the $\left(z_2 - \mu^0\right)$-subsidy intervention is optimal. 
\end{theorem}

\begin{proof}
The proof follows by induction on $k$. First, let us consider the base case when $k = 2$, and let $C$ be any cost that leads to convergence in two shots. Consider any starting point $\mu^0 \in [z_1,z_2]$. Note that the sequence of wealth $\mu^0 \to \mu^1 \to \mu^2$ must satisfy $\mu^1 < z_2$ and $\mu^2 \geq z_2$. Further, note that because $\mu^{t+1} = f\left(\mu^t + C \right) \leq \mu^t + C$ we must have $C \geq \mu^{t+1} - \mu^t$. We then have that the loss $L$ satisfies
\begin{align*}
L(C) & = \lambda C + (1-\lambda) \discnt (z_2 - \mu^1) + \discnt (\lambda C) \nonumber \\
& \geq \lambda (\mu^1 - \mu^0) + (1-\lambda) \discnt (z_2 - \mu^1) + \discnt [\lambda (z_2 - \mu^1)]\\ 
& =  \lambda (\mu^1 - \mu^0) + \discnt \left(z_2 - \mu^1\right)\\
& \geq \lambda (\mu^1 - \mu^0) + \lambda \left(z_2 - \mu^1\right) = \lambda(z_2 - \mu^0).
\end{align*}
This concludes the case of $k = 2$.

For $k > 2$, note that we have $\mu^{k-1} < z_2$ and $\mu^{k} \geq z_2$\juba{fixed}. Letting $C$ be any cost that leads to reaching $z_2$ in $k$ shots, we have that the loss function is given by
\begin{align*}
L(C) 
& \geq \lambda(\mu^1 - \mu^0) + (1-\lambda) \discnt (z_2 - \mu^1) + \left[\lambda \sum_{t = 1}^{k-1} \discnt^t C + (1 - \lambda) \sum_{t = 2}^{k-1} \discnt^t (z_2 - \mu^t)\right]
\\&= \lambda(\mu^1 - \mu^0) + (1-\lambda) \discnt (z_2 - \mu^1) + \discnt \left[\lambda \sum_{t = 1}^{k-1} \discnt^{t-1} C + (1 - \lambda) \sum_{t = 2}^{k-1} \discnt^{t-1} (z_2 - \mu^{t})\right]
\\&= \lambda(\mu^1 - \mu^0) + (1-\lambda) \discnt (z_2 - \mu^1) + \discnt \left[\lambda \sum_{t = 0}^{k-2} \discnt^t C + (1 - \lambda) \sum_{t = 1}^{k-2} \discnt^t (z_2 - \mu^{t+1})\right].
\end{align*}
The second term in the last line of the inequality is the loss function when starting at $\mu^1 \in [z_1,z_2]$ instead of $\mu^0$. Indeed, write $\nu^t = \mu^{t+1}$ the sequence that starts at $\mu^1$ and satisfies $\nu^k = \mu^{k-1} < z_2$ but $\nu^{k-1} = \mu^k \geq z_2$ (hence this new sequence converges in $k-1$ rather than $k$ steps); the loss of this sequence is given by:
\begin{align*}
&\lambda \sum_{t = 0}^{k-2} \discnt^t C + (1 - \lambda) \sum_{t = 1}^{k-2} \discnt^t (z_2 - \nu^t) = \lambda \sum_{t = 0}^{k-2} \discnt^t C + (1 - \lambda) \sum_{t = 1}^{k-2} \discnt^t (z_2 - \mu^{t+1}).
\end{align*}
By the induction hypothesis, since the cost of a one-shot intervention is lower than that of any $k-1$-shot intervention, we have that 
\[
\lambda \sum_{t = 0}^{k-2} \discnt^t C + (1 - \lambda) \sum_{t = 1}^{k-2} \discnt^t (z_2 - \mu^{t+1}) \geq \lambda \left(z_2 - \mu^1 \right).
\]
Therefore, $L(C)$ is lower bounded by
\begin{align*}
\lambda(\mu^1 - \mu^0) + (1-\lambda) \discnt (z_2 - \mu^1) + \discnt [\lambda (z_2 - \mu^1)] &\geq \lambda(\mu^1 - \mu^0) + \lambda (z_2 - \mu^1) \geq \lambda(z_2-\mu^0).
\end{align*}
\end{proof}

In particular, $1$-shot interventions become optimal when the discounting factor $\discnt$ is relatively large, or when $\lambda$ is relatively small. The first result intuitively arises because when $\discnt$ becomes large, the centralized designer cares about cost and wealth of the group at each time step; a $1$-shot intervention allows the designer to incur a single up-front cost for intervening (instead of inefficiently investing a smaller cost per round over more rounds, and losing some of this invested cost, up to $\Delta$, at each time step) while immediately reaching high wealth outcomes. On the other hand, no matter what $\discnt$ is, when $\lambda$ becomes small, the designer only cares about reaching high wealth as soon as possible, hence prefers faster interventions. We now provide sufficient conditions under which $1$-shot is not optimal:

\begin{theorem}
\label{thm:necessary}
If $\discnt < \lambda \left(1 - \frac{C}{z_2 - \mu^0}\right)$, the $C$-subsidy intervention has lower loss than the $1$-shot, $\left(z_2 - \mu^0\right)$-subsidy intervention.
\end{theorem}

\begin{proof}
Consider any intervention with cost $C$ such that $\mu^k \geq z_2$, i.e. we reach $z_2$ after at most $k$ time steps. First, remember that the loss for this intervention is given by
\[
\lambda \sum_{t=0}^{k-1} \discnt^{t} C + (1 - \lambda) \sum_{t=1}^{k - 1} \discnt^{t} (z_2 - \mu^{t}).
\]
Noting that $\mu^t \geq \mu^0$ for all $t$, hence $z_2 - \mu^t \leq z_2 - \mu^0$, we can upper bound the loss by\juba{simplified}
\begin{align*}
\lambda \sum_{t=0}^{+\infty} \discnt^{t} C + (1 - \lambda) \sum_{t=1}^{+\infty} \discnt^{t} (z_2 - \mu^0)
&= \frac{\lambda  C}{1 - \discnt} + (1-\lambda) \frac{\discnt}{1 - \discnt} (z_2 - \mu^0) 
\\& \leq \frac{\lambda  C}{1 - \discnt} + (1-\lambda) \frac{\rho }{1 - \discnt} (z_2 - \mu^0).
\end{align*}
In turn, we have that a sufficient condition for $C$-subsidy to have a lower loss than one-shot is given by
\[ 
\frac{1}{1-\discnt} \left( \lambda C + (1-\lambda) \discnt (z_2 - \mu^0) \right) < \lambda (z_2 - \mu^0).
\]
This can be rewritten as 
\[
\lambda C + (z_2 - \mu^0) \discnt - \lambda (z_2 - \mu^0) \discnt  < \lambda (z_2 - \mu^0) - \lambda (z_2 - \mu^0) \discnt,
\]
i.e.
\[
 (z_2 - \mu^0) \discnt < \lambda (z_2 - \mu^0) - \lambda C, 
\]
which immediately leads to the theorem statement. 
\end{proof}

Theorem \ref{thm:necessary} gives conditions under which the minimal 1-shot intervention has higher cost than the $C$-subsidy intervention. But recall that we can take $C$ as small as $\Delta+\epsilon$ (for arbitrarily small $\epsilon$) and still get an intervention that reaches the region of attraction for the highest wealth fixed point. Thus we have the following corollary, which gives a necessary condition for the 1-shot intervention to be optimal:

\begin{corollary}\label{cor: mincost_best}
If $\discnt <  \lambda \left(1 - \frac{\Delta}{z_2 - \mu^0}\right)$, then the $(\Delta + \varepsilon)$-subsidy intervention has lower loss  than the $1$-shot, $(z_2 - \mu^0)$-subsidy intervention as $\varepsilon \to 0$.
\end{corollary}

The above corollary provides the most stringent condition that we can derive from Theorem~\ref{thm:necessary} for $1$-shot not to be optimal. In particular, we note that the cheapest intervention we can use, the $(\Delta + \varepsilon)$-subsidy one, is better than the $1$-shot intervention so long as $\discnt <  \lambda \left(1 - \frac{\Delta}{z_2 - \mu^0}\right)$. We note that the combination of Theorem~\ref{thm: 1-shot_best} and Corollary~\ref{cor: mincost_best} show that when $\Delta$ becomes small and subsidy interventions are efficient, the condition that $\discnt \geq \lambda$ becomes nearly tight for optimality of the $1$-shot, $(z_2 - \mu^0)$-subsidy intervention. When $\Delta$ is large, there are still situations in which the condition of Corollary~\ref{cor: mincost_best} is essentially necessary and sufficient for the $(\Delta+\varepsilon)$-subsidy to be better than the 1-shot intervention, as evidenced by the example below:\juba{updated the discussion as well as the example below to focus on the fact that the condition of Corollary 1 is sometimes the right one/necessary and sufficient}.

\begin{example}
Suppose $f$ is continuous, but such that it is linear on interval $(a,b) \subset (z_1,z_2)$ with $f(x) = x - \Delta$ within said interval. We have immediately that $\mu^{t+1} = \mu^t + C - \Delta$ hence $\mu^t = \mu^0 + t (C - \Delta)$ so long as $t$ is such that $\mu^t \in (a,b)$. Here, the one-shot intervention still has loss $\lambda (z_2 - \mu^0)$. However, the $(\Delta + \epsilon)$-subsidy intervention reaches $z_2 - \varepsilon$, hence $z_2$, after no less than $T_\varepsilon = \frac{b - \mu^0}{\varepsilon} \to_{\varepsilon \to 0} +\infty$ time steps. In turn, it has loss at least
\begin{align*}
L(\Delta + \varepsilon) 
&\geq \lambda \sum_{t=0}^{T_\varepsilon - 1} \discnt^{t} (\Delta + \varepsilon) 
+ (1 - \lambda) \sum_{t=1}^{T_\varepsilon - 1} \discnt^{t} (z_2 - \mu^0 - t \varepsilon)
\\&\to_{\varepsilon \to 0}~~\frac{\lambda}{1 - \discnt} \Delta + \frac{\discnt(1 - \lambda)}{1 - \discnt} (z_2 - \mu^0).
\end{align*}
The proof of Theorem~\ref{thm:necessary} shows that the loss is also upper-bounded by
\[
L(\Delta + \varepsilon)  \leq \frac{\lambda}{1 - \discnt} \left(\Delta + \varepsilon\right) + \frac{\discnt(1 - \lambda)}{1 - \discnt} (z_2 - \mu^0).
\]
Hence, it must be that this bound is essentially tight, i.e. that
\[
L(\Delta + \varepsilon)  \to_{\varepsilon \to 0} \frac{\lambda}{1 - \discnt} \Delta + \frac{\discnt(1 - \lambda)}{1 - \discnt} (z_2 - \mu^0).
\]
In particular, in this case, the condition of Theorem~\ref{thm:necessary} and Corollary~\ref{cor: mincost_best} is not only sufficient but also necessary for the $\left(\Delta + \varepsilon\right)$-subsidy intervention to have better loss than the one-shot intervention.
\end{example}

\subsection*{Acknowledgements}
This work was supported in part by NSF grants AF-1763307 and FAI-2147212 and a grant from the Simons Foundation. 

\bibliographystyle{plainnat}
\bibliography{Arxiv/Arxiv_sourcefiles/refs.bib}

\appendix

\section{Omitted Proofs for Section~\ref{sec:wealth_dynamics}: Wealth Dynamics}\label{app:wealth_dynamics}
\subsection{Proof of Claim \ref{clm:bivar}}\label{app:bivar}

Because $S$ is a convex combination of $W$ and $T$, both $(T,S)$ and $(W,S)$ are multivariate Gaussians. The covariances are given by
\[
Cov(T,S) = \beta Cov(T,T) + (1 - \beta) Cov(T,W) = \beta \stdt^2,
\]
\begin{align*}
Cov(W,S) 
&= \beta Cov(W,T) + (1 - \beta) Cov(W,W) = (1 - \beta) \stdw^2
\end{align*}
and 
\begin{align*}
Cov(S,S) 
&= \beta^2 Cov (T,T) + (1-\beta)^2 Cov(W,W) + 2 \beta(1-\beta) Cov(T,W) 
= \beta^2 \stdt^2 + (1-\beta)^2 \stdw^2. 
\end{align*}

\subsection{Proof of Lemma~\ref{lem:update_rule}}\label{app: proof_update_rule}

Using Claim~\ref{clm:bivar}, we have that
\begin{align*}
     E[T|S=s]
     &=  \frac{Cov(T,S)}{Var(S)}(s -  (1-\beta)\mu)
     =  \frac{\beta \stdt^2}{\beta^2\stdt^2 + (1-\beta)^2\stdw^2} (s - (1-\beta)\mu),
\end{align*}
and 
\begin{align*}
     E[W|S=s] 
     &=  \mu + \frac{Cov(W,S)}{Var(S)}(s - (1-\beta)\mu)
    = \mu + \frac{(1-\beta)\stdw^2}{\beta^2\stdt^2 + (1-\beta)^2\stdw^2} (s -  (1-\beta)\mu).
\end{align*}
Therefore, the university admits a student with score $s$ if and only if 
\begin{align*}
\left(\frac{\alpha \beta \stdt^2 + (1 - \alpha) (1 - \beta) \stdw^2}{\beta^2\stdt^2 + (1-\beta)^2\stdw^2}  \right)(s - (1-\beta) \mu) \geq \thr - (1 - \alpha) \mu,
\end{align*}
which can be rewritten as
\begin{align*}
&\frac{s - (1-\beta) \mu}{\sqrt{\beta^2 \stdt^2 + (1 - \beta)^2 \stdw^2}} \geq \frac{\sqrt{\beta^2 \stdt^2 + (1 - \beta)^2 \stdw^2}}{\alpha \beta \stdt^2 + (1 - \alpha) (1 - \beta) \stdw^2} \cdot \left(\thr - (1 - \alpha) \mu\right).
\end{align*}
Noting that by Claim~\ref{clm:bivar}, $\frac{S - (1-\beta) \mu}{\sqrt{\beta^2 \stdt^2 + (1 - \beta)^2 \stdw^2}}$ follows a normal distribution with mean $0$ and variance $1$; the expression for $\mu^{t+1}$ (hence the update rule) is then given by
\begin{align*}
    &1 - \Phi\left(\frac{\sqrt{\beta^2 \stdt^2 + (1 - \beta)^2 \stdw^2}}{\alpha \beta \stdt^2 + (1 - \alpha) (1 - \beta) \stdw^2} \cdot \left(\thr - (1 - \alpha) \mu\right)\right)
    = 1 - \Phi\left(K\left(\alpha,\beta,\stdt,\stdw\right) \cdot \left(\thr - (1 - \alpha) \mu\right)\right)
\end{align*}
This concludes the proof.

\subsection{Proof of Claim~\ref{clm: f_shape}}\label{app: f_shape}

For simplicity of notations, let us write $K$ instead of $K(\alpha,\beta,\stdt,\stdw)$.  Continuity is immediate from $f$ being the composition of a linear (hence continuous) function and the continuous function $\Phi$. Now, we have
\[
f'(x) = \frac{K (1-\alpha)}{\sqrt{2 \pi}} \exp \left(- K^2 (\thr - (1-\alpha) x)^2/2\right) \geq 0,
\]
showing $f$ is increasing. Finally, the second order derivative of the update rule $f''(x)$ is given by

\[
\frac{K^3 (1-\alpha)^2 (\thr - (1-\alpha) x)}{\sqrt{2 \pi}} \cdot e^{- K^2 (\thr - (1-\alpha) x)^2/2}.
\]
The result immediately follows, as $f''(x) \geq 0$ if and only if $x \leq \frac{\thr}{1-\alpha}$.  

\subsection{Proof of Lemma~\ref{lem:fixed_pts}}\label{app:fixed_pts}

Let us write $g(x) = f(x) - x$. Note that $f(x)$ has a fixed point if and only if $g(x) = 0$.
\begin{enumerate}
    \item $\thr \geq 1 - \alpha$ and $f$ is convex on $[0,1]$. Then $g'(x) = f'(x) - 1$, $g''(x) = f''(x)$, and $g$ is also convex. Further, note that $g(0) = f(0) - 0 > 0$ and $g(1) = f(1) - 1 < 0$. Therefore, $g(x) = 0$ can only have one solution at most. Indeed, let $x^*$ be the smallest value in $[0,1]$ for which $g(x^*) = 0$; we have that for all $x \in (x^*,1]$, we can write $x^* = \lambda x + (1-\lambda) 1$ for some $\lambda \in (0,1]$, Then, we have $g(x) \leq \lambda g(x^*) + (1 - \lambda) g(1) < 0$ by convexity. 
    \item $\thr \leq 0$ and $f$ is concave on $[0,1]$. Then $f$ can have at most $1$ fixed point by the same argument as above. 
    
    \item 
    Otherwise, note that $g'(x) = f'(x) - 1$ is first increasing up until $x^* = \tau/(1-\alpha)$ then decreasing in $x$. Therefore, $g'$ has at most two zeros $x^-$ and $x^+$. If $g'$ has two zeros, they must satisfy $x^- < x^*$ and $x^+ > x^*$, and that $g'(x) < 0$ for $x < x^-$, $g'(x) \geq 0$ for $x \in [x^-,x^+]$, and $g'(x) < 0$ for $x > x_+$. $g$ then has at most three intersection with $0$, with the first intersection on $[0,x^-]$, the second on $[x^-,x^+]$, and the third on $[x^+,1]$. When $g'$ has at most one zero, $g$ can only have at most $2$ zeros
    
    
\end{enumerate}
This concludes the proof.

\section{Omitted Proofs for Section~\ref{sec:interventions}: Interventions for Long-term Fairness}\label{app:interventions}

\subsection{Proof of Theorem~\ref{thm:threshold}}\label{app:threshold}

This follows from the fact that $f(x,\thr)$ is decreasing in $\thr$ for all $x \in [0,1]$. First, this implies that $f(x,\thr') > f(x,\thr) \geq x$ for all $x \in [0,z_1(\thr)]$. Hence $z_1(\thr') > z_1(\thr)$. For the third fixed point, note that $f(z_3(\thr),\thr') > f(z_3(\thr),\thr) = z_3(\thr) $; because $f$ is continuous and $f(1) < 1$, this immediately implies that $f$ has a fixed point on $(z_3(\thr),1]$, hence $z_3(\thr') > z_3(\thr)$.

Finally, let us consider the case of the second fixed point. First, we note that it must be that $z_1(\thr') < z_2(\thr)$. Suppose this is not the case, it must be that $f(x,\thr') > x$ for all $x < z_2(\thr)$. Further, for all $x \in [z_2(\thr),z_3(\thr)]$, we must have $f(x,\thr') > f(x,\thr) \geq x$, hence $f(x,\thr') > x$ for all $x < z_3(\thr)$. This implies $z_1(\thr') > z_3(\thr)$. However, we must have $z_3(\thr) \geq \thr/(1-\alpha)$ while $z_1(\thr') \leq \thr'/(1-\alpha) < \thr/(1-\alpha)$, which is a contradiction. \\ Now that we have $z_1(\thr') < z_2(\thr)$, note that it must be that $f(x,\thr') < x$ on a small neighborhood $(z_1(\thr'),z_1(\thr')+ \varepsilon)$ by our characterization of the fixed points of $f$. Since $f$ is continuous and $f(z_2(\thr),\thr') > f(z_2(\thr),\thr) = z_2(\thr)$, there exists a fixed point on $(z_1(\thr'),z_2(\thr)$. Since $z_3(\thr') > z_3(\thr) > z_2(\thr)$, this must be the second fixed point $z_2(\thr')$. 

\subsection{Proof of Theorem~\ref{thm:beta}}\label{app:beta}

\juba{made significant changes, check}
The partial derivative of $f$ with respect to $\beta$ is given by
\begin{align*}
    &\pdv{f}{\beta}(x,\beta) = 
    [\thr-(1-\alpha)x]~\cdot~\phi\left(K(\alpha,\beta,\stdt,\stdw) (\tau - (1-\alpha) x)\right) 
    \\&\times \frac{(\alpha - \beta) \stdt^2 \stdw^2}{\sqrt{\beta^2 \stdt^2 + (1-\beta)^2 \stdw^2} (\alpha \beta \stdt^2 + (1-\alpha) (1 - \beta) \stdw^2)^2} 
\end{align*}
where $\phi$ is the probability density function of a standard Gaussian. Note that for $\alpha < \beta$, $\pdv{f}{\beta}(x,\beta) < 0$ when $x < \tau/(1-\alpha)$ and $\pdv{f}{\beta}(x,\beta) > 0$ when $x > \tau/(1-\alpha)$. In particular, $f(x,\beta') > f(x,\beta) \geq x$ for all $x \leq z_1(\beta) (< \tau/(1-\alpha))$, hence $f(.\beta')$ has no fixed point on $[0,z_1(\beta)]$. This means that $z_1(\beta') > z_1(\beta)$. Similarly, $f(x,\beta') < f(x,\beta) \leq x$ for all $x \geq z_3(\beta) (> \thr/(1-\alpha))$, hence $f$ has no fixed point on $[z_3(\beta),1]$ and $z_3(\beta') < z_3(\beta)$. A similar proof follows for $\alpha \geq \beta' > \beta$.

\end{document}